\documentclass{llncs}
\usepackage{times}
\usepackage{comment,amsfonts,amssymb,amsmath,graphicx,algorithmic,algorithm}

\ifx\pdftexversion\undefined
\usepackage[colorlinks,linkcolor=black,filecolor=black,citecolor=black,urlco
lor=black,pdfstartview=FitH]{hyperref}
\else
\usepackage[colorlinks,linkcolor=blue,filecolor=blue,citecolor=blue,urlcolor
=blue,pdfstartview=FitH]{hyperref}
\fi

\newcommand{\alert}[1]{\textbf{\color{red}
[[[#1]]]}\marginpar{\textbf{\color{red}**}}\typeout{ALERT:
\the\inputlineno: #1}}

\newcommand{\diam}{{\rm diam}}

\newcommand{\mommit}[1]{}
\newcommand{\namedref}[2]{\hyperref[#2]{#1~\ref*{#2}}}
\newcommand{\sectionref}[1]{\namedref{Section}{#1}}
\newcommand{\appendixref}[1]{\namedref{Appendix}{#1}}

\newcommand{\theoremref}[1]{\namedref{Theorem}{#1}}
\newcommand{\defref}[1]{\namedref{Definition}{#1}}

\newcommand{\algref}[1]{\namedref{Algorithm}{#1}}

\newcommand{\lemmaref}[1]{\namedref{Lemma}{#1}}

\usepackage{pdfsync}

\title{Space-Efficient Path-Reporting Approximate Distance Oracles}

\author{Michael Elkin\inst{1} \and Ofer Neiman\inst{1}\thanks{Supported in part by ISF grant No. (523/12) and by the European Union's Seventh Framework Programme (FP7/2007-2013) under grant agreement $n^\circ$303809.} \and Christian Wulff-Nilsen\inst{2}}
\institute{
Department of Computer Science, Ben-Gurion University of the Negev, Beer-Sheva, Israel. Email: \texttt{elkinm,neimano@cs.bgu.ac.il}.
\and
Department of Computer Science, University of Copenhagen. Email: \texttt{koolooz@di.ku.dk},
                  \texttt{http://www.diku.dk/$_{\widetilde{~}}$koolooz/}.
}

\begin{document}

\maketitle

\begin{abstract}
We consider approximate {\em path-reporting} distance oracles, distance labeling and labeled routing with extremely low space requirement, for general undirected graphs. For distance oracles, we show how to break the $n\log n$ space
bound of Thorup and Zwick if approximate {\em paths} rather than distances need to be reported%and if the graph is unweighted
.
For approximate distance labeling and labeled routing,
we break the previously best known space bound of $O(\log n)$ words per vertex%, again at the cost of larger stretch
. The cost for such space efficiency is an increased stretch. 
\end{abstract}

\section{Introduction}

%In this paper we consider succinct data structures for approximate distance estimation...

\subsection{Distance oracles}
Given a graph $G = (V,E)$ with edge weights, an {\em approximate distance oracle} for $G$ is a data structure that can report approximate distance queries between vertex
pairs efficiently. For any vertices $u,v\in V$, if $d_G(u,v)$ denotes the shortest path distance from $u$ to $v$ in $G$ and if
$\tilde{d}(u,v)$ denotes the approximate distance output by the oracle, we require that $d_G(u,v)\leq \tilde{d}(u,v)\leq\delta d_G(u,v)$, where $\delta\geq 1$
is the approximation (called also the (multiplicative) {\em stretch} parameter) of the oracle. The goal is to give an approximate distance oracle with small space, query time, stretch, and
(perhaps to a lesser extent) preprocessing time.

Our focus is on undirected graphs as it can be shown that no non-trivial oracles exist for directed graphs \cite{TZ01}. A seminal result in this area is that of Thorup and
Zwick \cite{TZ01}. For any positive integer $k$ and a graph with non-negative edge weights and with $m$ edges and $n$ vertices, they gave an approximate distance oracle with
space $O(kn^{1+1/k})$, stretch $2k-1$, query time $O(k)$, and preprocessing time $O(kmn^{1/k})$. For constant $k$, the trade-off between the first three parameters
is optimal, assuming a widely believed and partially proved girth conjecture of Erd\H{o}s \cite{E63}. For super-constant $k$, small improvements exist. In \cite{W13}, it was shown
how to improve the query time to $O(\log k)$ while keeping the same space, stretch, and preprocessing. More recently, Chechik \cite{C14} further improved this to $O(1)$ query time.
Mendel and Naor \cite{MN06} gave an oracle with $O(n^{1+1/k})$ space and $O(1)$ query time at the cost of a constant-factor increase in stretch.

So far, we have only discussed queries for approximate \emph{distances} but it is natural to require the data structure to also be able to report corresponding
\emph{paths}.
We say that an oracle is {\em path-reporting} if it can report those paths in time proportional to their lengths (in addition to the query time needed for distances), and we say that it is a not path-reporting oracle otherwise.
The oracle of \cite{TZ01} and the $O(\log k)$ query time oracle of \cite{W13} are path-reporting, but this is not the case for the oracle of Mendel and Naor \cite{MN06}
nor for the oracle of Chechik which uses their oracle as a black box. Note that a space requirement of order $kn^{1+1/k}$ is
$\Omega(n\log n)$ for any choice of $k$. In this paper, we focus on path-reporting distance oracles 
%for unweighted graphs 
that use $o(n \log n)$ space, albeit at the price of increased stretch.

\subsection{Distance Labeling}
A natural variant of distance oracles arises when we start to distribute the data structure among all vertices.
Consider a graph problem involving queries for pairs of vertices. In a \emph{labeling scheme} for this problem, the goal is to assign as short labels as possible to
each vertex of the input graph so that a query for any pair $(u,v)$ of vertices can be answered (preferably efficiently) exclusively from the labels assigned to $u$
and $v$. We are interested in a {\em distance labeling} scheme where given labels of two vertices $u$ and $v$, a distance estimate $\tilde{d}(u,v)$ that satisfies
$d_G(u,v) \le \tilde{d}(u,v) \le \delta \cdot d_G(u,v)$ can be efficiently computed.

Distance labeling was introduced in a pioneering work by Peleg \cite{P00a}. 
%(An even earlier implicit distance labeling scheme was given by Matousek \cite{Matousek96}.)  
The distance oracles of Thorup and Zwick \cite{TZ01} and their refinements \cite{W13,C14} can serve as distance labeling schemes as well. (The maximum label size becomes $O(n^{1/k}\cdot\log^{1-1/k}n)$ words, and other parameters stay intact.) This is, however, not the case for Mendel-Naor's oracle \cite{MN06}.
%Specifically, while the total size of that oracle is $O(n^{1+1/k})$, the number of trees that it uses is $O(k n^{1/k})$, and therefore, the maximum label size may be as large as $O(k n^{1/k})$ as well.

To summarize, {\em all} existing distance labeling schemes use $\Omega(\log n)$ words per label in the worst case. The labeling scheme that we devise in the current paper uses $o(\log n)$ words per label, for graphs with polynomially bounded diameter. On the other hand, its stretch guarantee is much larger than that of \cite{TZ01,W13,C14}.

\subsection{Labeled Routing}
In a closely related {\em labeled routing} problem we want to precompute two pieces of information for every vertex $u$ of the input graph. These are the {\em label} of $u$ and the {\em routing table} of $u$. Given a label of another vertex $v$, the vertex $u$ should decide to which neighbor $w$ of $u$  to forward a message  intended for $v$  based on its local routing table and on the label of $v$. Given this forwarded message with the label of $v$, the neighbor $w$ selects one of its own neighbors, and forwards it the message, and so forth. The {\em routing path} is the $u$-$v$-path which will eventually be taken by a message originated in $u$ and intended for $v$. (Assuming that the routing scheme is correct, the path will indeed end in $v$.) The {\em stretch} of a routing scheme is the maximum ratio between a length of a routing $u$-$v$ path and the distance $d_G(u,v)$ between $u$ and $v$, taken over all (ordered) pairs $(u,v)$ of vertices.

Labeled routing problem was introduced in a seminal paper by Peleg and Upfal \cite{PU88}, and it was studied in \cite{C99,EGP98,AP92,ANLP90}.
%All these references can be taken from TZ-SPAA01 paper
The state-of-the-art labeled routing scheme was devised by Thorup and Zwick \cite{TZ01a}. It provides stretch $4k-5$ and uses routing tables of size $O(\mathit{polylog}(n) \cdot n^{1/k})$ and labels of size $O(k \cdot {{\log^2 n} \over {\log\log n}})$.

The space usage by current routing schemes is at least logarithmic in $n$ (counted in words; each word is $O(\log n)$ bits). In many settings such space requirement is prohibitively large. In this paper we show a labeled routing scheme in which the space requirement per vertex (both labels and routing tables) can be as small as one wishes, for graphs with diameter at most some polynomial in $n$. On the other hand, similarly to the situation with distance labeling schemes, the stretch guarantee of our scheme is much larger than that of \cite{TZ01a}.

\subsection{Our Results}
We introduce two new data structures that report paths in undirected graphs. All have query time proportional to the length of the returned path.
The first applies to weighted graphs with diameter polynomially bounded
in $n$. For any $t\geq 1$, it reports paths of stretch $O(\sqrt tn^{2/\sqrt t})$ using space $O(tn)$. It may be distributed as a labeling scheme using $O(t)$ space per vertex (or $O(t\log n)$ bits), and the preprocessing time is $O(tm)$. See \theoremref{thm:soc} for the formal statement. \footnote{For arbitrary diameter $\Delta$, the space and preprocessing time increase by a factor of $O(\log_n\Delta)$.}
This data structure can also be modified to provide labeled routing. Specifically, using tables of size $O(t)$ and labels of size $O(\sqrt{t})$ our routing scheme provides stretch $O(\sqrt{t} \cdot n^{2/\sqrt{t}} \cdot \log n)$.
%See Theorem \ref{thm:soc} for the formal statement.

The second data structure is a distance oracle that applies only to unweighted graphs. In one of the possible settings, it can provide for any parameters $k\ge 1$ and $\epsilon>0$, a path-reporting distance oracle with stretch $O(kn^{1/k}\cdot(k+n^{\epsilon/k}))$, using space $O(kn/\epsilon)$ and preprocessing time $O(kmn^{1/k})$. See \theoremref{thm:multi}, and also \theoremref{thm:oracle} for more possible tradeoffs.
%Alternatively, allowing additive stretch, our oracle has space $O(kn)$ and stretch $(O(k^2n^{1/k}),O(kn^{2/k}))$, where for
%$\alpha\geq 1$ and $\beta\geq 0$, we say that a distance estimate $\tilde d_G(u,v)$ has stretch $(\alpha,\beta)$ if $d_G(u,v)\leq\tilde d_G(u,v)\leq\alpha d_G(u,v) + \beta$.

To our knowledge, our distance labeling and labeled routing schemes are the first that use $o(\log n)$ words per vertex. Our distance oracles are the first path-reporting oracles for general graphs that use space $o(n \log n)$.

%\alert{mention routing, excluded minor?}

\subsection{Overview of Techniques.}  
Our first oracle is based on a collection of sparse covers. Roughly speaking, a sparse cover for radius $\rho$ has two parameters: $\beta$ is the {\em radius blow-up}, and $s$ the {\em overlap}. The cover is a collection of clusters, each of diameter at most $\beta\rho$, such that every ball of radius $\rho$ is fully contained in at least one cluster, and every vertex is contained in at most $s$ clusters (see \defref{def:cover} below for a formal definition). Sparse covers were introduced
by \cite{AP90}, and found numerous applications in distributed algorithms and routing (see, e.g. \cite{PU88,AP90a,P93,AP95,AGMNT08}). For the application to distance oracles and labeling schemes,
the radius blow-up corresponds to stretch and the overlap to space. The standard construction of \cite{AP90} for parameter $k\ge 1$ has radius blow-up $k$ and overlap $O(kn^{1/k})$. This overlap is at least $\Omega(\log n)$, and translates to such space usage per vertex. Here we show that one can obtain the inverse parameters: radius blow-up $O(kn^{1/k})$ with overlap $2k$ (in fact we can obtain overlap $(1+\epsilon)k$ for any fixed $\epsilon>0$).

Our first construction of a distance labeling scheme is very simple: it uses a collection of such sparse covers for all distance scales, and maintains a shortest-path tree for each cluster. In order to answer a path query, one needs to find an appropriate cluster in the right scale, and return a path from the corresponding tree.

Our second data structure combines sparse covers with a variation on the Thorup-Zwick (TZ) distance oracle. In order to save space, the "bunches" of the TZ oracle are kept only for a small set of carefully selected vertices. Furthermore, the TZ trees (from which the path is obtained) are pruned to contain only few important vertices. Given a path query, our pruned TZ oracle can only report a "skeleton" of the approximate shortest
path in the original graph. This skeleton contains few vertices (roughly one vertex per $p$ steps, for some parameter $p$). We then use a sparse cover to "fill in" the gaps in the path, which induces additional stretch.

\begin{comment}
\subsection{Additional Results - Minor Free Graphs}

Strong-diameter neighborhood covers for  graphs that exclude minor $K_r$, for a parameter $r$, were given in \cite{AGMW} and \cite{Busch}.
The covers of Abraham et al. \cite{AGMW} have stretch $O(r^2)$ and overlap $2^{O(r)} \cdot r!$.
The covers by Busch et al. \cite{Busch} have stretch 4 and overlap $f(r) \log n$, where $f()$ is a very fastly-growing function originated from Robertson-Seymour's structure theorem.
We devise a construction of covers with stretch  $k \ge r^2$, for a parameter $k$, and overlap $O(\log n/\log (k/r^2))$.
Our covers are better than those due to \cite{AGMW} for $r \ge \log\log n/\log\log\log n$.  Similarly, they compare favorably to the covers of  \cite{Busch} when $r$ is sufficiently large.

As a corollary of our construction of covers for minor-free graphs we also obtained new distance labeling and labeled routing schemes.
There are known constructions of distance oracles for such graphs, due to Abraham and Gavoille \cite{AG06}, and due to Kawabarayashi, Klein and Sommer \cite{KKS}. The latter schemes produce distance estimates with stretch guarantee $1+ \epsilon$, for an arbotrarily small $\epsilon > 0$, but their query time is $O(f(r) \cdot \log n )$  and $O(f(r) \cdot \log^2 n)$, respectively.
We get a distance labeling scheme with stretch $O(1)$, space per label $O(\log^2 n \cdot 2^{O(r)})$ and query time $O(\log (k\log_n \Delta))$, where $\Delta$ is the weighted diameter of the graph.

\end{comment}

\subsection{Related Work}
%\alert{fill in...  we should mention that Agarwal's et. al. result is not really path-reporting for stretch $>$3}
There has been a large body of work on distance oracles, labeling and routing for certain graph families (planar, excluded-minor, etc.) and bounded doubling dimension metrics \cite{T01,HM06,AG06,KKS11}. In these settings the stretch factor is usually $1+\epsilon$, which cannot be obtained with $o(n^2)$ space for general graphs.

For sparse graphs, very compact distance oracles were recently devised by Agarwal et al. \cite{AGH11,AGH12}.
They devise two types of distance oracles. One of them has small stretch but requires large space. This distance oracle is indeed path-reporting, but due to their large space requirement they are irrelevant to the current discussion. The other type of distance oracles in \cite{AGH11} has stretch at least 3. These latter distance oracles are very sparse, but they are not path-reporting. \footnote{The paper erroneously claims that they are \cite{AGH-p}.}

\subsection{Organization of the Paper}

After some basic definitions in \sectionref{sec:pre}, we introduce sparse covers with small overlap in \sectionref{sec:cover}. Our first data structure for weighted graphs with diameter polynomially bounded
in $n$, is presented in \sectionref{sec:simple}, and its adaptation for compact routing in \appendixref{sec:route}. The second data structure with improved parameters for unweighted graphs is given in \sectionref{sec:coverTZ}.
%Finally, in \sectionref{sec:minor} we discuss improved results for graphs excluding a minor.

\section{Preliminaries}\label{sec:pre}
Let $G=(V,E)$ be an undirected weighted graph, with the usual shortest path metric $d_G$. We always assume the minimal distance in $G$ is 1. For a subset $U\subseteq V$ let $G[U]$ denote the induced graph on $U$.
%For $U\subseteq V$ and radius $r$ let $B(U,r)=\{v\in V\mid d_G(v,U)\le r\}$. By $|G|$ we mean the number of vertices in $G$.
For a parameter $\rho>0$, and two sets of balls ${\cal B}, {\cal S}\subseteq\{B(v,\rho)\mid v\in V\}$,
define $\partial_{\cal B}({\cal S})=\{B\in{\cal B}\mid \exists S\in{\cal S},~B\cap S\neq\emptyset\}$ to be the subset of balls from ${\cal B}$ that intersect with a ball from ${\cal S}$.

\begin{definition}\label{def:cover}
A collection of clusters ${\cal C}=\{C_1,\dots,C_t\}$ is called a strong diameter {\em $(\beta,s,\rho)$-sparse cover} if
\begin{itemize}
\item {\em Radius blow-up:} $\diam(G[C_i])\le\beta\rho$ for all $i\in[t]$.
\item {\em Padding:} For each $v\in V$, there exists $i\in[t]$ such that $B(v,\rho)\subseteq C_i$.
\item {\em Overlap:} For each $v\in V$, there are at most $s$ clusters in ${\cal C}$ that contain $v$.
\end{itemize}
%We say that the cover has {\em strong diameter} $\beta \rho$ if $\diam(G[C_i])\le\beta\rho$ for all $i\in [t]$.
%We say that it has a {\em weak diameter} $\beta \rho$ if for every $i$ and $u,v \in C_i$, $d_G(u,v) \le \beta \rho$.
For a vertex $v$ and a cluster $C_i$ such that $B(v,\rho) \subseteq C_i$, we say that the vertex $v$ is {\em padded} by the cluster $C_i$.
\end{definition}

\section{Sparse Covers with Small Overlap}
\label{sec:cover}

%Fix a radius $\rho>0$, and consider the set of balls $B(v,\rho)$ for all $v\in V$.

%\alert{Say something about Awerbuch-Peleg, and the inverse parameters we show here}

In this section we show how to construct a sparse cover with arbitrarily low overlap. Our construction  essentially inverts the parameters in the classical tradeoff of \cite{AP90}, which has low radius blow-up. We use a region growing technique on the set of balls of radius $\rho$.

\begin{algorithm}
\caption{$\texttt{Sparse-Cover}(G,\rho,k)$}\label{alg:cover}
\begin{algorithmic}[1]
\STATE ${\cal C}=\emptyset$.
\STATE ${\cal U}=\{B(v,\rho)\mid v\in V(G)\}$.
\WHILE {${\cal U}\neq\emptyset$}
\STATE ${\cal R}={\cal U}$.
\WHILE {${\cal R}\neq\emptyset$}
\STATE Let $B\in {\cal R}$.
\STATE Let ${\cal S}=\{B\}$.
\WHILE {$|\partial_{\cal R}({\cal S})| \ge |{\cal S}|\cdot\left(1+\frac{\log n}{k\cdot n^{1/k}}\right)$}
\STATE ${\cal S}\leftarrow \partial_{\cal R}({\cal S})$.
\ENDWHILE
\STATE ${\cal C}\leftarrow {\cal C}\cup \{ \bigcup_{B' \in {\cal S}} B'\}$.
\STATE /* A new cluster  $C = \bigcup_{B' \in {\cal S}} B'$ is added to ${\cal C}$. */
\STATE ${\cal R}\leftarrow {\cal R}\setminus \partial_{\cal R}({\cal S})$.
\STATE ${\cal U}\leftarrow {\cal U}\setminus {\cal S}$.
\ENDWHILE
\ENDWHILE
\end{algorithmic}
\end{algorithm}

\begin{theorem}
For any weighted graph $G$ on $n$ vertices, any $\rho>0$ and $k\ge 1$, there exists a strong diameter $(8k\cdot n^{1/k},2k,\rho)$-sparse cover.
\end{theorem}

\begin{proof}
Consider \algref{alg:cover} for creating a sparse cover.  Observe that we only throw a ball from ${\cal U}$ when it is contained in ${\cal S}$ and will surely be contained in a cluster. Thus when the algorithm terminates all $\rho$-balls are padded.

Let $n_i$ denote the number of balls in ${\cal U}$ at the end of the $i$-th iteration of the outer loop. Then $n_0=n$, and by the termination condition of the while loop on line 8,
\[
n_{i+1}< n_i\cdot\left(\frac{\log n}{k\cdot n^{1/k}}\right)~.
\]
This implies that
\[
n_{2k}< n\cdot\left(\frac{\log n}{k\cdot n^{1/k}}\right)^{2k}=\frac{1}{n}\cdot\left(\frac{\log n}{k}\right)^{2k}\le 1~,
\]
where the last inequality holds because the function $((\log n)/k)^{2k}$ is maximal when $k=(\log n)/2$, in which case it is $n$. We conclude that the algorithm terminates after at most $2k$ phases. When forming a cluster $C = \bigcup_{B \in {\cal S}} B$, all balls in $\partial_{\cal R}({\cal S})$ (the balls that intersect $C$) are removed from ${\cal R}$ and thus will not be considered in the current phase (the loop starting at line 5), which implies that every point $v\in V(G)$ can belong to at most a single cluster per phase. So the total overlap is at most $2k$.

It remains to bound the strong diameter of any cluster. A cluster ${\cal S}$ starts as a ball of diameter at most $2\rho$, and in each iteration of line 8 its size (number of balls it contains) increases by a factor of at least $\left(1+\frac{\log n}{k\cdot n^{1/k}}\right)$. After $2k\cdot n^{1/k}$ iterations its size will be at least
\begin{equation}
\label{eq:num_iter}
\left(1+\frac{\log n}{k\cdot n^{1/k}}\right)^{2k\cdot n^{1/k}}>e^{\log n}> n~.
\end{equation}
For the last inequality we used that $1+x> e^{x/2}$ when $0<x\le 1$ (indeed $(\log n)/(k\cdot n^{1/k})\le 1$).
Inequality \eqref{eq:num_iter} is a contradiction, so the number of iterations in line 8 is less than $2k\cdot n^{1/k}$. In each such iteration the diameter can increase by at most $4\rho$, so the total diameter is bounded by $8k\cdot n^{1/k}\cdot\rho$.
\end{proof}

\paragraph{Remark:} A similar algorithm and calculation shows that for any $1/k\le\epsilon\le 1$, one can obtain also a $(8k\cdot n^{1/k}/\epsilon,(1+\epsilon)k,\rho)$-sparse cover, though we shall not require this generalization here.
In \appendixref{sec:fast} we show a fast construction of such sparse covers, with slightly worse constants.

\begin{comment}
\section{Path Reporting Low Space Distance Oracles}\label{sec:simple}

An (approximate) distance oracle for a graph $G=(V,E)$ is a data structure that (approximately) supports distance queries. We would like our oracles to also report the paths in $G$ that correspond to the distance estimation given. There are several parameters of interest in an oracle, such as the pre-processing time, the approximation factor (the stretch), the space usage, and the query time.

The Thorup-Zwick distance oracles can provide a $2k-1$ stretch oracle with $O(k\cdot n^{1+1/k})$ space, such that if a path of length $l$ is returned, then the time is $O(l)$. Note that the space is at least $n\log n$ for any choice of $k$. Our goal is to obtain oracles with smaller space.
The oracle of Mendel and Naor \cite{MN06} provides distance {\em estimates} with stretch $O(k)$ and uses space $O(n^{1+1/k})$, but it does not report approximate shortest {\em paths}.
%\comment
%The oracle of Mendel and Naor which has stretch $O(k)$ and space $O(n^{1/k})$ (with $O(1)$ query time) seems to achieve this, but it is not known how to implement it in a way that supports path queries.
%\commentend

%One of our results builds an oracle with multiplicative and additive stretch. An oracle has $(\alpha,\beta)$-stretch if for all $u,v\in V$, the estimated distance $\tilde{d}$ satisfies $d_G(u,v)\le\tilde{d}(u,v)\le \alpha\cdot d_G(u,v)+\beta$.

\end{comment}

\section{Small Space Distance Labeling Scheme}\label{sec:simple}

In this section we provide a distance labeling scheme, that can also serve as a path-reporting distance oracle, which is built from a collection of sparse covers. Its parameters are somewhat inferior to the parameters of the distance oracle from \sectionref{sec:coverTZ}. On the other hand, the latter construction applies only to unweighted graphs, while the construction is this section applies to weighted graphs. Also, it is not clear to us if the construction of \sectionref{sec:coverTZ} can be converted into a distance labeling scheme.

\begin{theorem}\label{thm:soc}
For any weighted graph $G=(V,E)$ on $n$ vertices with diameter $\Delta$, and any $t\ge 1$, there exists a distance labeling scheme with stretch $O(\sqrt{t}\cdot n^{2/\sqrt{t}})$ using $O(t\cdot\log_n\Delta)$ space (or $O(t\cdot\log\Delta)$ bits) per vertex, that can be constructed in $O(t|E|\cdot\log_n\Delta)$ time. Furthermore, this data structure can also serve as a path-reporting distance oracle, whose query time is proportional to the length of the returned path, plus $O(\log (t \cdot \log_n\Delta))$.
\end{theorem}
{\bf Remark:} Observe that when the diameter $\Delta$ is at most polynomial in $n$, the required space is $O(t)$ words per vertex.
\begin{proof}
Fix a parameter $1\le k$, and let $\Delta=\diam(G)$, $\gamma=8kn^{1/k}$, $q=\lceil \log_{n^{1/k}} \Delta \rceil = \lceil k \log_n \Delta \rceil$. For each $i\in\{0,1,\dots,q\}$ create a $(\gamma,2k,n^{i/k})$-sparse cover ${\cal C}_i$. For each cluster $C\in{\cal C}_i$ choose an arbitrary shortest path tree (SPT) spanning $G[C]$. Every vertex stores a hash table containing the names of the SPTs it is contained in, and for each such tree the vertex only needs to store a pointer to its parent in the tree and the distance to the root of the tree. Since every vertex is contained in at most $2k$ clusters per level, the total space used is $O(k\cdot q)$ per vertex. Observe that if $\Delta=\text{poly}(n)$ then $q=O(k)$.
In addition, for every vertex $u\in V$ and $i\in[q]$, store a pointer to the SPT of a cluster $C_i(u)\in {\cal C}_i$ such that $B(u,n^{i/k})\subseteq C_i(u)$.

Next we describe an algorithm for
answering  a path query between $u,v\in V$.  Let $i\in[q]$ be such that $n^{(i-1)/k}\le d_G(u,v)< n^{i/k}$.
Let $J = \{j \in [0,q] : v \in C_j(u)\}$ be the set of indices $j$ such that $v \in C_j(u)$. By the padding property of the sparse cover, $v \in B(u,n^{j/k}) \subseteq C_j(u)$, for every $j \ge i$. Hence every index $j \ge i$ belongs to $J$. We will conduct a binary search on $[0,q]$ to find an index $j$ such that $j \in J$ and $j-1 \not \in J$. (Alternatively, we will discover that $0 \in J$.) By the above considerations the index $j$ that we will find satisfies $j \le i$.
As $u$ holds a pointer to $C_j(u)$ for every $j\in\{0,1,\dots,q\}$ and $v$ stores the names of clusters containing it in a hash table, deciding if $j \in J$ requires $O(1)$ time.  Next, both $u,v$ follow the path to the root in the SPT created for $C_j(u)$. By taking a step towards the root in the path with the longer remaining distance, we can guarantee that the paths will meet at the least common ancestor of $u,v$.
%The time invested in each step is proportional to the diameter of the cluster $C_j(u)$, which is at most $\gamma^{j+1}$. The total time is $\sum_{j=0}^iO(\gamma^{j+1})=O(\gamma^{i+1})=O(\gamma^2\cdot d_G(u,v))$.

The query time is bounded by length of the returned path, which is $O(\diam(C_i(u)) = O(\gamma n^{i/k})=O(\gamma n^{1/k} \cdot d_G(u,v))$, so the stretch is $O(\gamma n^{1/k})=O(k n^{2/k})$.
In addition we spend $O(\log q) = O(\log (k \log_n \Delta))$ time for the binary search.
% (the additive term of $O(\log k)$ for the binary search is swallowed in the stretch
If one is willing to settle for $\gamma=64kn^{1/k}$ (rather than $8k n^{1/k}$), then the pre-processing expected time is $O(qk\cdot |E|)$, using the construction of \sectionref{sec:fast}. 

Finally, note that for the labeling scheme, we can find the appropriate $j\in J$ and return the sum of distances from $u,v$ to the root of the SPT of $C_j(u)$, just by inspecting the labels of $u,v$.

\end{proof}

\section{Small Space Path-Reporting Distance Oracles}
\label{sec:coverTZ}
In this section we show a path-reporting distance oracle with improved stretch, at the price of being applicable only for unweighted graphs. Also we do not know
if it is possible
 to distribute the information among vertices, i.e., to convert this oracle into a labeling scheme.
The distance oracle in this section has both additive and multiplicative stretch. For $\alpha\geq 1$ and $\beta\geq 0$, we say that a distance estimate $\tilde d$ has {\em $(\alpha,\beta)$-stretch} if for all $u,v\in V$, $d_G(u,v)\leq\tilde d(u,v)\leq\alpha \cdot d_G(u,v) + \beta$.

\begin{theorem}\label{thm:oracle}
For any unweighted graph $G=(V,E)$ on $n$ vertices, any integers $k,p,t\ge 1$, there exists a path-reporting distance oracle with $\left(O(t\cdot kn^{1/k}),O(p\cdot kn^{1/k})\right)$-stretch, using $O(kn +tn^{1+1/t}/p)$ space. Furthermore, the query time is proportional to the length of the returned path.
The oracle can be constructed in $O(tmn^{1/t})$ time.
%\alert{+$O(\log k)$ with Christian's result? or use Chechik's new result for $O(1)$? Is it even clear we can apply these results here (we don't have all the bunches..)?}
\end{theorem}

\paragraph{Proof Overview:} Fix a parameter $p$.  and we partition the distances to those smaller than $p$ and those larger.
In order to be space efficient, we "prune out" most of the vertices in the distance oracle of Thorup-Zwick. We will choose a subset $N\subseteq V$, of size $n/p$, that touches the (approximately) $p$-neighborhood of any vertex of $V$. The TZ-oracle will be responsible for the large distances between any two vertices in $N$: it should be able to report a sufficiently dense "skeleton" of an approximate shortest path.  All consecutive distances on the path are roughly $p$. We show that one can significantly reduce the size of each of the TZ trees, while still maintaining this usability. We augment our data structure with a sparse cover that will handle all the small distances: specifically we need to "fill in" the paths between consecutive vertices in the skeleton, and the paths between each vertex to its representative in $N$.

\subsection{Construction}
\label{sec:construction}

We shall use the following Lemmata. The first one is folklore.

\begin{lemma}\label{lem:hit}
For every unweighted graph on $n$ vertices and parameter $r$, there is a set of at most $2n/r$ vertices that intersects every ball of radius $r$.
\end{lemma}

The next lemma can be found, e.g., in \cite{NS07}, Lemma 12.1.5.

\begin{lemma}\label{lem:separator}
For every tree $T$ on $n$ vertices and parameter $r$, there is a set of at most $2n/r$ vertices whose removal separates $T$ into components of size at most $r$ each.
\end{lemma}

One of the building blocks of our oracle is a variation of the Thorup-Zwick oracle \cite{TZ01}.
We briefly recall the TZ construction with stretch parameter $t$: Define $A_0=V$ and for each $i\ge 1$ sample $A_i$ from $A_{i-1}$ by including every element of $A_{i-1}$ independently with probability $n^{-1/t}$. Finally, set $A_t=\emptyset$. For $u\in V$ define the {\em bunch} of $u$ as $B(u)=\{w\in A_{i-1}\mid d_G(u,w)<d_G(u,A_i),i\ge 1\}$.  For each $w\in V$, if $i$ is such that $w\in A_{i-1}\setminus A_i$,  define $C(w)=\{u\in V\mid d_G(w,u)<d_G(u,A_i)\}$. (Note that the cluster $C(w)$ contains all vertices $u$ such that $w\in B(u)$).
It can be shown that
for every $u \in C(w)$, all vertices on the shortest path between $u$ and $w$ also belong to $C(w)$ as well. As a result, an SPT for $C(w)$ is a subtree of an SPT rooted at $w$ for the entire graph $G$.
%$C(w)$ is a subtree of an SPT rooted at $w$ in $G$.
 During the preprocessing  such a tree spanning $C(w)$ is created for each $w\in V$. We denote it by $T_w$.
%Create a SPT tree $T_w$ for each cluster $C(w)$ rooted at $w$.

In the original data structure, each vertex $u\in V$ stored the vertices in $B(u)$ and their distances from $u$ in $G$. For each $i=0,\dots,t-1$, it also stored the special vertex $p_i(u)\in A_i$, which is the closest vertex to $u$ in $A_i$. The query algorithm on $u,v$ uses only the information stored by the query vertices to produce some $w\in B(u)\cap B(v)$ such that $d_G(u,v)\le d_G(u,w)+d_G(w,v)\le (2k-1)d_G(u,v)$,
%, and it also holds that $w=p_i(u)$ or $w=p_i(v)$ for some $i$.
and the actual path could be obtained from $T_w$. It is also shown in \cite{TZ01} that for each $v\in V$, the expected size of $B(v)$ is $O(kn^{1/k})$.

As we aim to save space, we will only store the bunches $B(u)$ for a few vertices. Fix a parameter $p$, and let $N$ be a set of size $n/p$ that hits every ball of radius $2p$. (See  \lemmaref{lem:hit},  $r=2p$.) Only the vertices $v\in N$ will store the bunches $B(v)$ and special vertices $p_i(v)$ of the TZ-oracle.
Since the total size of the trees $\{T_w\}_{w\in V}$ is equal to the total size of the bunches, we also need to prune these trees.
For each $w\in V$, let $R_w$ be the set given by \lemmaref{lem:separator} applied on the tree $T_w$ with $r=p$, of size at most $2|T_w|/p$, and let $R=\cup_{w\in V}R_w$. Let $\bar{T}_w$ be the pruned tree that contains only the vertices of $T_w$ that are in $R_w\cup N\cup\{w\}$. Specifically, each vertex in the pruned tree $\bar{T}_w$ will store a pointer to its nearest ancestor which is also in $\bar{T}_w$, the distance to it, and the distance to the root.

We shall also require a sparse cover (as constructed in \sectionref{sec:cover}). For a parameter $k\ge 1$, let ${\cal C}$ be a $(8kn^{1/k},2k,3p)$-sparse cover, and for each cluster $D\in{\cal C}$ create an SPT spanning $G[D]$. As before, in each tree a vertex stores a pointer to its parent and the distance to the root. Additionally, every vertex $u\in V$ stores a hash table of trees containing it, a pointer to $D(u)$, a cluster in which it is padded, and a pointer to some $u'\in N$ such that $d_G(u,u')\le 2p$. Note that $u'\in D(u)$.

{\bf Remark:} Observe that we build the data structure on all of $V$ and then prune the obtained TZ trees, rather than applying the TZ structure restricted to the vertices of $N$ (which would seem an obvious simplification). This is because the TZ trees that will be produced from the metric induced on $N$ may have arbitrarily large weights on the edges.
%In such a case the cover would be useless in completing the "skeleton" paths.
One then would need a different mechanism for replacing these edges by paths of the original graph. This is because the cover ${\cal  C}$ can only help filling in gaps of length up to $3 p$.

\paragraph{Bounding the Size of the Oracle:} Next we show that the space used by our oracle is $O(kn+tn^{1+1/t}/p)$. To see this, note that since
the cover overlap is $2k$,
to store the SPTs of the cover and the relevant pointers for each $u\in V$ requires only $O(kn)$ space.  Next, we bound the size of the stored bunches. The number of vertices in $N$ is at most $n/p$, and since the expected bunch size for each vertex is $O(tn^{1/t})$, the total (expected) size of the bunches we store is $O(tn^{1+1/t}/p)$. It remains to bound the size of the pruned trees. Since every vertex $u\in N$ is expected to appear in $O(tn^{1/t})$ trees (the number of trees equals its bunch size), the contribution of vertices in $N$ to the size of the trees $\{\bar{T}_w\}_{w\in V}$ is again $O(tn^{1+1/t}/p)$. Finally, recall that the (expected) size of all the trees $\{T_w\}_{w\in V}$ is $O(tn^{1+1/t})$.  \lemmaref{lem:separator} implies that in each tree only fraction of $2/p$ of the vertices are in $R$ (rounded up), thus the contribution of vertices in $R$ to the pruned trees is $O(n+tn^{1+1/t}/p)$. The roots of the pruned trees contribute only $O(n)$ to the size.

\paragraph{Construction Time:} The bottleneck in our construction time is to find the clusters $C(w)$. With the construction of Thorup and Zwick, we get a bound of
$O(tmn^{1/t})$.

\subsection{Answering Path Queries}

In order to answer a path query on $u,v\in V$, we first check if $v\in D(u)$. If so, taking the paths to the root in the SPT created from $D(u)$ from both $u$ and $v$, as done in \sectionref{sec:simple}, will give a path of length $O(\diam(D(u))=O(kn^{1/k}\cdot p)$, which induces such additive stretch.
%\alert{I wasn't sure we need Interval Tree Routing for this - is there an issue with just following pointers?}
%(In order to implement this, one can use Interval Tree Routing in each of the SPT trees created from clusters $D$, where each vertex stores the $2k$ interval labels for the trees containing it).

Note that if $d_G(u,v)\le 2p$ it must be that $v\in D(u)$, so the complementary  case is when $d_G(u,v)> 2p$. We shall use the pruned TZ data structure in the following way. First use the pointers stored at $u,v$ to find $u'\in D(u)\cap N$ and $v'\in D(v)\cap N$. Using the TZ algorithm, which only requires the information stored by $u'$ and $v'$, we find $w\in B(u')\cap B(v')$ with $d_G(u',w)+d_G(v',w)\le (2t-1)\cdot d_G(u',v')$. Since $w$ is contained in both bunches, and $u',v'\in N$, we get that $u',v'\in\bar{T}_w$. Since $\bar{T}_w$ is also a shortest path tree from $w$,
\[
d_{\bar{T}_w}(u',v') \le d_G(u',w)+d_G(v',w)\le(2t-1)d_G(u',v')~.
\]
%Note that the algorithm of \claimref{claim:tz} uses only $B(u')$ and $B(v')$ with the special points $p_i(u'),p_i(v')$, and thus it can be executed using the information stored in our data structure.

The "skeleton path" $u'=u_0,u_1,\dots u_{l'}=v'$ induced by the (pruned) tree $\bar{T}_w$ from $u'$ to $v'$ has stretch $2t-1$. It can be obtained efficiently by following paths towards the root $w$ from $u',v'$, as done above. Our goal now is to show that there is a subpath, in which all consecutive distances are in the range $[p,3p]$, these "gaps" will be covered by the sparse cover. Since removing $R_w$ partitions $T_w$ into subtrees of size at most $p$, it cannot be the case that there is a path in $T_w$ of length $p$ that does not intersect $R_w$ (such a path induced a subtree with $p+1$ vertices). We conclude that for each $j\in[l']$, $d_G(u_{j-1},u_j)\le p$.
We further prune this skeleton path, to get a sub-path in which all consecutive distances are in the range $[p,3p]$. This can be achieved by greedily deleting excessive points (those closer than $p$ to the last point we kept) while traversing the path, and making sure to keep both $u',v'$. It is not hard to verify that the maximum distance between consecutive points will be at most $3p$.
\begin{comment}
To see why this works, suppose $w$ is the
current vertex in the traversal. If its distance to $v'$ is at least $p$ then its distance to the previously kept vertex must be less than $2p$; otherwise the predecessor
of $w$ should have been kept. If the distance from $w$ to $v'$ is less than $p$ then the last chosen vertex $w'$ has distance at most $2p$ to $w$; otherwise, there would
be a vertex between $w'$ and $w$ with distance at least $p$ to $w'$ and to $w$ and hence also to $v'$. Thus, $w'$ will be the last vertex kept before reaching $v'$ and
its distance to $v'$ is at most $3p$. \alert{should we clarify more on that? I made a silly mistake in the previous iteration.. CWN: I hope this clarified it.}. 
\end{comment}
Let $u'=v_0,v_1,\dots,v_l=v'$ be the resulting skeleton path.

For each $j\in[l]$ find a path in $G$ from $v_{j-1}$ to $v_j$ using the sparse cover. Since $d_G(v_{j-1},v_j)\le 3p$ we get that $v_j\in D(v_{j-1})$. So we can obtain a path in $G$ from $v_{j-1}$ to $v_j$ of length at most $O(p\cdot kn^{1/k})$, in the same manner we handled the base case above (where $v\in D(u)$). Note that this induces a $O(kn^{1/k})$ stretch for each $j$ (because $d_G(v_{j-1},v_j)\ge p$), so the final multiplicative stretch is $O(t\cdot kn^{1/k})$. In a similar manner, since both $d_G(u,u'),d_G(v,v')\le 2p$, we obtain from the sparse cover paths from $u$ to $u'$ and from $v'$ to $v$ of distance at most $O(p\cdot kn^{1/k})$. The latter contributes to the additive stretch.

The running time of the query is proportional to the length of the path returned, since after finding the tree $\bar{T}_w$, we just follow pointers to the roots in both the pruned TZ-trees and in the SPT of the cover, in constant time per step. Note that the $O(t)$ time to find $T_w$ is dominated by the stretch factor which we can assume
is bounded by the length of the path.
%The total expected size of our data structure is $O(kn+t\cdot n^{1+1/t}/p)$.
%Choosing $t=k$ and $p=n^{1/k}$ yields a distance oracle with $O(kn)$ space and $\left(O(k^2n^{1/k}),O(kn^{2/k})\right)$-stretch, for any value of $k$.
This proves \theoremref{thm:oracle}.

\subsection{Improved Multiplicative Stretch using Several Covers}

%If one desires a pure multiplicative stretch, the trivial guarantee is $O(kn^{1/k}\cdot (p+t))$.

Choosing $t=k$ and $p=n^{1/k}$ in the parameters of \theoremref{thm:oracle} yields stretch of $(O(k^2n^{1/k}), O(k n^{2/k}))$.
In terms of purely multiplicative stretch  it is $O(k n^{1/k} \cdot (k + n^{1/k}))$.
 Next we show how to improve one of the factor of $n^{1/k}$ at the cost of increased space. Instead of a single cover, we use a collection of $s$ sparse covers, and obtain the following theorem.
\begin{theorem}\label{thm:multi}
For any unweighted $n$-vertex graph $G = (V,E)$, any positive integer parameter $k$, and any parameter $\epsilon > 0$, there exists a path-reporting distance oracle with space $O(kn/\epsilon)$ and stretch $O(kn^{1/k}\cdot(k+n^{\epsilon/k}))$. Furthermore, the query time is proportional to the length of the returned path.
\end{theorem}
\begin{proof}
For each $i\in[s]$ let ${\cal C}_i$ be a $(8kn^{1/k},2k,(3p)^{i/s})$-sparse cover, and for each cover store the same information per vertex as above. For $i\in[s]$, denote by $D_i(u)$ the cluster in ${\cal C}_i$ in which $u$ is padded. Recall that the (additive) factor of $O(p\cdot kn^{1/k})$ in the stretch was inflicted in the base case when $d_G(u,v)\le 2p$, and also from completing the path from $u$ to $u'$ and from $v$ to $v'$.

Given some $u,v\in V$ with the guarantee that $d_G(u,v)\le 3p$, we can find an index $i$ such that $v\in D_i(u)$ and $v\notin D_{i-1}(u)$ (or that $v\in D_1(u)$) by binary search. The path between $u,v$ in the SPT induced from $D_i(u)$ is of length at most $O(p^{i/s}\cdot kn^{1/k})$, and can be found
in the same way as was described above.
Since $v\notin D_{i-1}(u)$, $d_G(u,v)\ge (3p)^{(i-1)/s}$, and thus the stretch factor is only $O(p^{1/s}\cdot kn^{1/k})$.
Combining this with the stretch factor of $O(t\cdot kn^{1/k})$ on the path from $u'$ to $v'$, we get total stretch $O((t+p^{1/s})\cdot kn^{1/k})$. Note that we only use the collection of $s$ covers twice per query. Specifically, all the skeleton missing paths will be filled in using the cover ${\cal C}_s$ as before.
(The cover ${\cal C}_s$  has exactly the same parameters as the cover ${\cal C}$ from \sectionref{sec:construction}.)
So the additive $O(\log s)$ term for the query time is surely dominated by the stretch.
Choosing $t=k$, $p=n^{1/k}$ and $s  = \lceil 1/\epsilon \rceil$ completes the proof.
\end{proof}

\section{Conclusions}
We gave space-efficient approximate distance oracles, distance labeling, and labeled routing for undirected graphs. Our distance oracles break the $n\log n$ space bound
of Thorup and Zwick for unweighted graphs and can report approximate shortest paths in time proportional to their length. The cost is an increase in
(multiplicative and/or additive) stretch. For distance labeling and routing, we break the previously best known space bound of order $\log n$ words at the cost of larger
stretch.

It might be possible to improve preprocessing of our distance oracles, e.g., by using techniques from~\cite{W12} for graphs that are not too sparse. Note that
the oracle of Mendel and Naor achieves linear space and logarithmic stretch but it can only report approximate distances, not paths. We state it as an open problem
whether a path-reporting oracle with linear space and polylogarithmic stretch exists which reports a path in time proportional to its length.
%\alert{maybe mention open problems, like linear space with polylog stretch? CWN: more is needed here, e.g., for labeling and routing.}

%%%%%%%%%%%%%%%%%%%%%%%%%%%%%%%%%%%%%%%%%%%%%%%%%%%%%%%%%%%%%%%%%%%%%%%%%%
\bibliographystyle{alpha}
\bibliography{oracle}
%%%%%%%%%%%%%%%%%%%%%%%%%%%%%%%%%%%%%%%%%%%%%%%%%%%%%%%%%%%%%%%%%%%%%%%%%%

\appendix

\section{Fast Construction of Sparse Covers}\label{sec:fast}
For a weighted graph $G=(V,E)$, we show a probabilistic construction of $(64k\cdot n^{1/k},2k,\rho)$-sparse cover in time $O(k\cdot|E|)$,
for any $\rho>0$. The main building block are padded partitions. A {\em partition} $P=\{C_1,\dots,C_t\}$ of the graph $G$ is a collection of pairwise disjoint clusters whose union covers $V$. We say that the partition
is strong diameter $\Delta$-bounded if $\diam(G[C_i])\le\Delta$.
For a partition $P$ and a vertex $x$, let $P(x)$ denote the cluster of $P$ that contains $x$.
We use the following Lemma that appears (implicitly) in \cite{B96}.
\begin{lemma}\label{lem:padd}
For any weighted graph on $n$ vertices there exists a distribution ${\cal P}$ over strong diameter $\Delta$-bounded partitions, such that for all $v\in V$ and $0\le\beta\le 1/8$,
\[
\Pr_{P\sim{\cal P}}[B(x,\beta\Delta)\subseteq P(x)]\ge n^{-16\beta}~.
\]
Furthermore, one can sample from this distribution in linear time.
\end{lemma}

In order to construct a cover, just sample a partition according to the distribution of \lemmaref{lem:padd} for $2k$ times, with parameters $\Delta =64k\cdot n^{1/k}\cdot\rho$ and $\beta=\rho/\Delta$, and return the collection of clusters obtained.
The radius bound on each cluster is $\Delta$, and since each partition consists of disjoint clusters, each point will be covered exactly $2k$ times. The probability that a certain ball of radius $\rho=\beta\Delta$ is not contained in any of the $2k$ partitions is at most
\[
\left(1-n^{-16\beta}\right)^{2k}=\left(1-e^{-16\ln n/(64kn^{1/k})}\right)^{2k}\le\left(\frac{\log n}{2kn^{1/k}}\right)^{2k}<1/n^{3/2}~,
\]
which holds since $(\log n/(2k))^{2k}\le n^{1/2}$.
Using a union bound over all $n$ balls, there is high probability that each of them will be contained in some cluster.
As each partition is created in linear time, the total running time is $O(|E|\cdot k)$.

\section{Routing}\label{sec:route}

We consider a compact routing framework, in which every vertex in the graph has a short label (word size), and stores a routing table. Given a vertex $u$ and a label of $v$, using the routing tables starting at $u$ and given only the label of $v$, we should route from $u$ to $v$ quickly. Specifically, we show the following result.

\begin{theorem}\label{thm:route}
Fix any parameter $k$. Any weighted graph $G=(V,E)$ on $n$ vertices with diameter $\Delta$ admits a compact routing scheme, in which the labels are of size $O(k\log_n\Delta)$ and the routing tables are  of size $O(k^2\log_n\Delta)$. For any two vertices $u,v\in V$, the scheme produces routing paths from
 $u$ to $v$ of length  at most $O(kn^{2/k}\log n\cdot d_G(u,v))$.
\end{theorem}

We shall use Interval Tree Routing described in \cite{P00},  lemma 26.1.2.
\begin{theorem}\label{thm:tree-route}
Let $T$ be a tree on $n$ vertices with depth $d$, then there exists a compact routing scheme that uses a single word ($O(\log n)$ bits) as a label and produces paths of length $O(d\log n)$.
\end{theorem}

Using the same framework as \sectionref{sec:simple}, one can extend the distance labels described there to a compact routing scheme with almost the same parameters: we only lose a factor of $O(\log n)$ in the routing time. Each vertex $u$ will have a {\em routing table} of size $O(k\cdot q)$.
 Specifically, for each level $i\in[q]$ and each SPT containing $u$ in this level, the vertex will store the relevant information required for interval-tree routing (see \theoremref{thm:tree-route}). The {\em label} of $u$ will be much shorter, of size $O(q)$: for every $i\in[q]$ store the information only for the SPT created from $C_i(u)$, i.e., the cluster in which $u$ is padded. In order to route from $u$ to $v$, we find an index $i$ such that $v\in C_i(u)$ but $v\notin C_{i-1}(u)$ (this can be done since we have all the information for $u$) and route in the corresponding SPT using interval-tree routing (\theoremref{thm:tree-route}). Recall that $v\notin C_{i-1}(u)$ implies that $d_G(u,v)\ge n^{(i-1)/k}$.

The depth of the SPT is bounded by the diameter of the cluster $C_i(v)$, and $\diam(C_i(v))\le \gamma\cdot n^{i/k}=O(\gamma n^{1/k}\cdot d_G(u,v))$.
(Recall that $\gamma = 8 k n^{1/k}$.)
 So the the length of the routing path  in the tree is $O(\gamma n^{1/k}\log n\cdot d_G(u,v))=O(kn^{2/k}\log n\cdot d_G(u,v))$.

\end{document}